\documentclass[10pt,onecolumn,draftclsnofoot]{IEEEtran}

\usepackage[T1]{fontenc}

\usepackage{cite}

\usepackage{amsmath}

\usepackage{mathtools}
\usepackage{array}%
\usepackage{amsthm}
\usepackage{amsfonts}
\usepackage{amssymb}
\usepackage{xcolor}

\newtheorem{theorem}{Theorem}

\newtheorem{claim}{Claim}

\newtheorem{corollary}{Corollary}

\usepackage{scalerel}
\usepackage{authblk}
\usepackage{tikz}
\usetikzlibrary{svg.path}

\definecolor{orcidlogocol}{HTML}{A6CE39}
\tikzset{
  blaaa/.pic={
    \fill[orcidlogocol] 
                 svg{M256,128c0,70.7-57.3,128-128,128C57.3,256,0,198.7,0,128C0,57.3,57.3,0,128,0C198.7,0,256,57.3,256,128z};
    \fill[white] 
                 svg{M86.3,186.2H70.9V79.1h15.4v48.4V186.2z}
                 svg{M108.9,79.1h41.6c39.6,0,57,28.3,57,53.6c0,27.5-21.5,53.6-56.8,53.6h-41.8V79.1z M124.3,172.4h24.5c34.9,0,42.9-26.5,42.9-39.7c0-21.5-13.7-39.7-43.7-39.7h-23.7V172.4z}
                 svg{M88.7,56.8c0,5.5-4.5,10.1-10.1,10.1c-5.6,0-10.1-4.6-10.1-10.1c0-5.6,4.5-10.1,10.1-10.1C84.2,46.7,88.7,51.3,88.7,56.8z};
  }
}

\newcommand\orcidicon[1]{\href{https://orcid.org/#1}{\mbox{
\begin{tikzpicture}[remember picture]
\coordinate(A);
\coordinate(B) at ($(A)-(3pt,-10pt)$);
\end{tikzpicture}
\begin{tikzpicture}[overlay,remember picture,yscale=-0.035,xscale=0.035,transform shape]
\pic at (B){blaaa};
\end{tikzpicture}
}{}}}
\usepackage{hyperref} 

\begin{document}

\title{Optimal Wireless Caching with Placement Cost}

\author{{
\large{Yousef AlHassoun\protect\orcidicon{0000-0002-3165-1821}, {\em Student Member, IEEE}, 
Faisal Alotaibi\protect\orcidicon{0000-0001-8024-1776}, {\em Student Member, IEEE}, 
\\Aly El Gamal\protect\orcidicon{0000-0002-0400-4506}, {\em Senior Member, IEEE}, 
and Hesham El Gamal\protect\orcidicon{0000-0002-0900-5962}, {\em Fellow, IEEE}}
}
\thanks{Yousef AlHassoun, Faisal Alotaibi, and Hesham El Gamal are with the ECE Department of the Ohio State University, Columbus, OH (e-mail: \{alhassoun.1, alotaibi.12, elgamal.2\}@ece.osu.edu). Aly El Gamal is with the ECE Department of Purdue University, West Lafayette, IN (e-mail: elgamala@purdue.edu)}}

\usetikzlibrary{backgrounds,calc,fit} 
\maketitle
\begin{abstract}
Coded caching has been shown to result in significant throughput gains, but its gains were proved only by assuming a placement phase with no transmission cost. A \emph{free} placement phase is, however, an unrealistic assumption that could stand as an obstacle towards delivering the promise of coded caching. In \cite{y2019jointly}, we relaxed this assumption by introducing a general caching framework that captures transmission costs for both delivery and placement phases under general assumptions on varying network architectures, memory constraints, and traffic patterns. Here, we leverage this general framework and focus on analyzing the effect of the placement communication cost on the overall throughput and the structure of the optimal caching scheme, under the assumptions of the worst case traffic pattern and unlimited memory at the end users.
Interestingly, we find relevant network configurations where uncoded caching is the optimal solution. 
\end{abstract}


\IEEEpeerreviewmaketitle

\section{Introduction}

\IEEEPARstart{R}{}ising demand from bandwidth-intensive applications, such as virtual reality and video sharing, is imposing a significant burden on the available wireless infrastructure. This can negatively impact - in a significant manner - the Quality of Service (QoS) guarantees during peak times. On the other hand, several reports have shown that the infrastructure is largely underutilized during {\em off-peak} hours, as the network is specifically designed to handle the peak demand \cite{federal2002fcc,adelstein2003facilitating}. This observation suggested the need for caching schemes that {\em balance the traffic} over peak and off-peak times.


In this work, we focus on leveraging the multi-casting advantage of coded caching in wireless networks to reduce the cost of data delivery (see e.g., \cite{Maddah-Ali2014FundamentalCaching, Maddah-Ali2015DecentralizedTradeoff,yu2018exact}). In our model, there are two distinct phases of communication between the source, i.e., Service Provider (SP), and destinations, i.e., end-users. In the {\em placement} phase, the SP sends information to be stored in each end user's local storage. Note that communication costs are assumed to be zero in the current literature during the placement phase (i.e., free communication). In the {\em delivery} phase, multi-cast coding is carefully implemented to minimize total delivery costs by exploiting the wireless channel's broadcast nature. Here we relax the assumption of zero placement cost and adopt several components of the optimization theoretic framework in \cite{jin2017structural} and \cite{daniel2017optimization} to derive optimal caching schemes in a more realistic setting.


We address this problem by jointly optimizing the placement and delivery strategies under more realistic assumptions on the cost of communication. We study the case when the memory constraint is relaxed, and focus on the impact of the communication cost during the placement phase on the optimal scheme. Moreover, we allow the overall cost of communication to be measured through two distinctive factors between the two phases. First, our model contains a different cost-per-transmission scaling factor for each phase. Second, we allow a different network architecture in the placement phase while we adopt the shared medium assumption in the delivery phase, resulting in a generalized model that enables capturing the case of varied networks (e.g., Wi-Fi versus cellular) in the two phases.  We then consider minimizing the cost of transmission in the delivery phase, while ensuring that the cost of transmission in the placement phase is not exceeded. This approach is used to derive valuable insights into the design of optimal schemes and to create a more precise interpretation of the relative advantages that can be leveraged in practical settings from coded multi-cast caching.

We first proposed this general framework in~\cite{y2019jointly}. Here we derive rigorous results on the structure of the optimal scheme and resulting overall gain under the assumptions of the worst case traffic pattern and unlimited caching memory at the end users.

\section{System Model}\label{Sec:Sys_model}
We consider a service provider (SP) who supplies $N \in \mathbb{N}$ files to $K \in \mathbb{N}$ users through a shared error-free link, over two phases namely, the placement phase and the delivery phase, where $K \leq N$. Let $\mathcal{N} \triangleq \{1, 2, ...,N\}$ and $\mathcal{K} \triangleq \{1, 2, ...,K\}$ denote the sets of file and user indices, respectively. We denote the file (word) with index $n \in \mathcal{N}$ by $\boldsymbol{W}_n$, whereas we assume that every user $k\in\mathcal{K}$ has an unlimited cache size.

\textbf{Placement cost function:}
In order to transmit to $r \in \mathcal{K}$ users during the placement phase, the SP incurs a cost 
\begin{align}\label{ct}
        c_r = \rho r^{\alpha}, 
    \end{align}
where $\alpha \in [0,1]$ is a cost parameter to capture the varying network architecture, and $\rho \in [0,1]$ is a linear cost multiplier. 
The special case of a shared medium in the placement phase can be recovered from this model by setting $\alpha=0$, whereas the special case of a Time Division Multiple Access (TDMA) channel used for placement corresponds to $\alpha=1$. 

Let $R_o$ (off-peak) and $R_p$ (peak) be the rates of transmission that the SP incurs during the placement and delivery phases, respectively. In order to minimize the peak time rate, the SP facilitates caching during the placement phase. In particular, with a slight abuse of notation, we say that the SP will transmit a message $\boldsymbol{Y_o}$ to all users, and every user $k$ will store a part $\boldsymbol{Z_k}$ in its memory. Note that, depending on the off-peak network architecture, this transmission of $\boldsymbol{Y_o}$ can be a single broadcast transmission, multiple unicast or multicast transmissions, or a combination thereof.  During the delivery phase, the SP will receive $K$ user requests.  Let $D_k$ denote the index of the file requested by user $k$ during the delivery phase, and let $\boldsymbol{D}=(D_1,...,D_K)$ denote the sequence of all user requests. The SP then broadcasts a multicast message $\boldsymbol{Y_p}$ over the shared link, which will be used for reconstructing the requested message with the aid of the cached content at each user.

\section{Problem Formulation} \label{Sec:Problem_for}

\subsection{Placement Phase}
The SP partitions each file $\boldsymbol{W}_n$ into $|\mathcal{P}(\mathcal{K})|$ non-overlapping subfiles, where $\mathcal{P}(\mathcal{K})$ is the powerset of the total number of users; these subfiles are denoted by $\boldsymbol{W}_{n,S}$, where $S \in \mathcal{P}(\mathcal{K})$. Moreover, the subfiles are classified according to their types $t \in \{0,1,2,...,K\}$, where $t = |S|$. Examples of subfiles for file $\boldsymbol{W}_n$, if we have $K=3$ would be $\boldsymbol{W}_{n,\{1,2\}}$ with type $t = 2$, and $\boldsymbol{W}_{n,\{1,2,3\}}$ with type $t=3$. The type will help us to group and label subfiles according to their role in the placement process. 

After this division, the SP will transmit in the placement phase (off-peak time) a concatenated message $\boldsymbol{Y}_o$ of all subfiles of type $t>0$ for all files to all users, that is:
 \begin{equation}
     \boldsymbol{Y}_o = \Big(\boldsymbol{W}_{n,S}: n\in \mathcal{N}, S\in \mathcal{P}(\mathcal{K})\backslash \emptyset \Big),
 \end{equation}
 where $\emptyset$ is the empty set.
Upon receiving this broadcast message, every user $k$ will store subfiles that have its index in the subfile label, that is:
\begin{equation}
    \boldsymbol{Z}_k = \Big(\boldsymbol{W}_{n,S}: n\in\mathcal{N}, S\in \mathcal{P}(\mathcal{K})\backslash \emptyset, k \in S\Big). 
\end{equation}
Let $x_{n,S}$ denote the ratio between the number of bits in $\boldsymbol{W}_{n,S}$ and the number of bits in $\boldsymbol{W}_n$. More precisely, we have: 
\begin{equation}
x_{n,S} = \dfrac{|\boldsymbol{W}_{n,S}|}{|\boldsymbol{W}_n|} \leq 1. 
\end{equation}
Also, let the file partitioning vector parameter be 
\begin{equation}
\boldsymbol{x}=\Big(x_{n,S}:\quad n\in\mathcal{N}, S\in \mathcal{P}(\mathcal{K})\Big).    
\end{equation}
We will restrict our attention to the special case of the problem with uniform file lengths, $|\boldsymbol{W}_n| =|\boldsymbol{W}|, \forall n\in \mathcal{N}$.  Furthermore, we will be following the simplification in \cite{jin2017structural} and \cite{daniel2017optimization} under the worst case scenario. For file $n$, we will have that all subfiles with the same type $t$ have the same size:
\begin{align}
    x_{n,S} &= x_t, \quad \forall n\in \mathcal{N}, S\in \mathcal{P}(\mathcal{K}): |S|=t.   \label{simp1}
\end{align}
We note that $\boldsymbol{x}$ should satisfy the following constraints: 
\begin{align}
& \sum^{K}_{t=0} a_t x_t = 1, \label{eq2} \\
& 0 \leq x_t , \quad \forall  t\in \{0,1,...,K\}, \label{eq1}
\end{align}
where $a_t = \binom{K}{t}$. Here, (\ref{eq2}) represents the file partitioning constraint. Therefore, 
the rate in the placement phase given all system parameters is as follow:
\begin{equation}\label{gRo}
R_o(\boldsymbol{x})= N \sum ^{K}_{t=1} a_t c_t x_t. 
\end{equation}

\subsection{Delivery Phase}
Upon receiving the user requests $\boldsymbol{D}$, the SP distinguishes the needed files and the possible opportunities for simultaneously serving multiple requests for global gain using a coded broadcast message. 
For each subset $S \in {\cal P}({\cal K})\backslash \emptyset$  where $|S|=t$, note that every $t-1$ users with indices in $S$ share a subfile stored in their memory, and it is needed by the remaining user in $S$. More precisely, for any $k \in S$, the subfile $\boldsymbol{W}_{D_k,S\backslash \{k\}}$ is missing at the memory of user $k$, whereas it is present in the memory of any user in $S\backslash \{k\}$. For example, if we have $K=3$ and user $k=1$ requests file $D_1$, the subfile $\boldsymbol{W}_{D_1,\{2,3\}}$ is missing at user $k=1$, while it is available at users $k=2,3$.  The transmitted coded multicasting message in the delivery phase (peak time) can hence be described as:
\begin{equation}
     \boldsymbol{Y}_p = \Big(\oplus_{ k \in S}\  \boldsymbol{W}_{D_k,S\backslash \{k\}}: S\in \mathcal{P}(\mathcal{K})\backslash \emptyset\Big),
\end{equation}
 where $\oplus$ denotes the bitwise XOR operation. 
 We focus on the \textit{worst case scenario}, where the SP will receive $K$ \textit{different} requests. Therefore, the rate in the delivery phase is given by:
 
\begin{equation}\label{Rp}
    R_p(\boldsymbol{x}) =  \sum_{t=0}^{K-1} a_t b_t x_t,
\end{equation}
where $b_t = \frac{K-t}{t+1}$.

\subsection{Optimization Problem}
To avoid minimizing one rate (peak time rate) at the cost of the other (off peak time rate), we introduce a condition on the placement phase rate that guarantees that no other peak (congestion) is created at the placement time. That is:
\begin{align}
    R_o(\boldsymbol{x}) \leq R_p(\boldsymbol{x}). \label{rate}
\end{align}
The SP defines its optimization cost function as minimizing the peak time rate while taking into account not to exceed the off peak time rate. Then, the SP optimization problem is:
\begin{align}
    & \underset{\boldsymbol{x}}{\text{minimize}}
    \quad  R_p(\boldsymbol{x})  \label{Opt1}\\ 
    & \text{subject to}
    \quad (\ref{eq2}), (\ref{eq1}), (\ref{rate})  \label{Opt111}.
\end{align}

\section{Main Result}
\begin{theorem}\label{the1}
For the worst case scenario, the optimal placement phase will only have at most two types of subfiles, and the optimal caching decision for type $t$ is given by:
\begin{equation*}
    x_t^* =
    \begin{cases*}
    1/a_t,       & \text{if } $\rho = \gamma_t$,\quad and \quad $ \alpha \geq \sigma_{t}$,\\
    P_1,       & \text{if } $\rho > \gamma_{t}$, \quad and \quad $\sigma_{t} \leq \alpha < \sigma_{t-1}$,\\  
    P_2,      & \text{if } $\gamma_t < \rho < \gamma_{t-1}$, \quad and \quad $ \alpha \geq \sigma_{t-1}$,\\ 
    P_3,    & \text{if } $\gamma_{t+1} < \rho < \gamma_{t}$, \quad and \quad $  \alpha \geq \sigma_{t}$,\\
    0,   & otherwise, \\
    \end{cases*}
\end{equation*}
where $\gamma_t = \dfrac{K-t}{t^\alpha (t+1) N}$, $\forall t \in {\cal K}$, and $\gamma_0 = 1$, $\sigma_{t} =\log_{\frac{t+1}{t}} \left( \dfrac{t+1}{t+2} \right)+1$, $\forall t \in {\cal K}\backslash\{K\}$, $\sigma_{0} =1$, and $\sigma_K=0$,\\$P_1 = \dfrac{K (t+1)}{ a_t \left(N  c_t (t+1) + t(K +1)\right)}$, $P_2 = \dfrac{(t+1)(t-K-1) + c_{t-1}N t(t+1) }{a_t \left(Nt(t+1) (c_{t-1} - c_{t}) - (K+1)\right)} $, \\and $P_3 = \dfrac{(t+1)(K-t-1) - c_{t+1}N (t+1)(t+2) }{a_t \left(N(t+1)(t+2) (c_t - c_{t+1}) - (K+1)\right)}$.
\end{theorem}
{\bf Proof:} We note that the optimization problem described in (\ref{Opt1}) and (\ref{Opt111}) is a Linear Programming (LP) problem. 
We start by introducing a new variable $y_t$ which can be described as the ratio between the total number of bits used for a single type $t$ and the total number of bits in the file. More precisely, 
\begin{equation}
   y_t = a_t x_t, \quad \forall t \in \mathcal{K}.
\end{equation}
Moreover, we denote $y_0$ as the fraction occupied by the reactive part of the file, which is not being cached at any user during the placement phase and will only be sent during the delivery phase. From constraint \eqref{eq2}, we have 
\begin{equation}\label{y0}
    y_0 = 1 - \sum \limits_{t=1}^{K} y_t.
\end{equation}
We use \eqref{y0} to substitute $y_0$ and represent our optimization problem in terms of the caching variables \{$ y_t, \quad t \in \mathcal{K}$\}.

Hence, we reformulate the original problem after simple manipulation as:
\begin{align}
    &\max_{\{y_t\}_{t=1}^{K}}  \qquad   \sum^{K}_{t=1}\dfrac{t}{(t+1)} y_{t}, \label{lp11}\\ 
    &\text{subject to} \quad  \sum_{t=1}^{K}q_t y_t\leq 1, \label{lp22}\\
    & \quad\qquad\qquad  \sum^{K}_{t=1} y_t \leq  1,  \label{lp33}\\
    & \qquad\qquad\quad 0 \leq y_t,    \quad\forall t \in {\cal K},\label{lp44}
\end{align}
where $q_t = \dfrac{c_tN(t+1)+t(K+1)}{K(t+1)} $. The optimization problem here is also a Linear Programming (LP) problem where $\{y_t, t \in {\cal K}\}$  are our decision variables. Aside from the non-negativity constraint of \eqref{lp44}, we have only two constraints \eqref{lp22} and \eqref{lp33} that limit our feasible region. From the Fundamental Theorem of Linear Programming \cite{lin1schrijver1998theory,lin2vanderbei2015linear, lin3matousek2007understanding}, the maximum number of non-zero variables in an optimal solution for the optimization problem is the same as the number of linearly independent columns of the constraints coefficient matrix, which is at most two in this case. The problem can hence be simplified by having only two caching types of subfiles with non-zero size.
Now, we rewrite our problem in terms of two variables $i,j \in \mathcal{K}$. Then, we have:
 \begin{align}
 &\max_{y_i, y_j, i, j}  \quad \dfrac{i}{(i+1)} y_{i} + \dfrac{j}{(j+1)} y_{j}, \label{s1}\\ 
    &\text{subject to} \quad  q_i y_{i} + q_j y_{j} \leq 1,\label{s2}\\
    & \quad\qquad\qquad    y_{i} + y_{j} \leq 1 , \label{s3}\\
    & \quad\qquad\qquad y_i, y_j \geq 0. \label{s4}
\end{align}


Since $\alpha\leq 1$, we note that the first derivative of $q_t$ is positive and the second derivative is negative. Therefore, $q_t$ is a concave increasing function in $t$ resulting in the following cases, which will characterize the solution to our LP based on the parameters ($\rho, \alpha$).
\subsection{Cost limited regime}
We identify this case when \eqref{s2} is the only constraint that limits the feasible region of our optimization problem, that is when $\left\{ q_t > 1, \quad  \forall t \in \mathcal{K}\right\}$. A necessary and sufficient condition for this case to occur based on the analysis of $q_t$ is:
\begin{equation}
     q_1  > 1,
     \end{equation}
     or equivalently:
     \begin{equation}
    \rho  > \dfrac{K-1}{2N}.     
     \end{equation}
    
Under this condition, our optimization problem becomes to find the optimal type $t$ given only one constraint. The corner points from this single constraint are:
 \begin{equation}
    \left( y_t = \dfrac{K(t+1)}{N c_t (t+1)+ t(K+1)},\left(y_{\tilde{t}}=0, \forall \tilde{t} \neq t\right) \right), \quad \forall t \in \mathcal{K},\label{yt}
 \end{equation}
 and the optimization problem for this case is:
 \begin{align}
 &\max_{t \in \mathcal{K}} \quad \dfrac{t}{(t+1)} \dfrac{K(t+1)}{N c_t (t+1)+ t(K+1)}.
\end{align}
The optimal type $t$ for this case can be identified based on the following range of $\alpha$ (see \cite{y2019jointly} for details):
\begin{equation}
    \sigma_{t} \leq \alpha < \sigma_{t-1}. \label{alphahat}
\end{equation}

\subsection{Free placement regime} 
We identify this case when \eqref{s3} is the only constraint that limits the feasible region of our optimization problem, that is when $\left\{ q_t \leq 1, \quad  \forall t \in \mathcal{K}\right\}$. A necessary and sufficient condition for this case to occur is:
\begin{equation}
      q_K  \leq 1,\\
      \end{equation}
      or equivalently,
      \begin{equation}
    \rho  = 0.
\end{equation}
The corner points from the single constraint in this case are given by $\left(y_t = 1, \left(y_{\tilde{t}}=0, \forall \tilde{t} \neq t\right)\right), \forall t \in \mathcal{K}$, and the optimization problem for this case is:
 \begin{align}\label{case2}
 &\max_{t \in \mathcal{K}} \quad \dfrac{t}{(t+1)}.
\end{align}

As the objective function described in \eqref{case2} is increasing with $t$, the optimal caching type when $\rho = 0$ is $t = K$ (uncoded delivery) regardless of the value of $\alpha$.
\subsection{Architecture limited regime} We identify this case when  $\{ q_t \leq 1, \quad  \forall t \in \{1, .. , a\}\}$ while $\{ q_t > 1, \quad  \forall t \in \{b, .. , K\}\}$, where $b=a+1\leq K$. The necessary and sufficient condition for this case to occur is:
\begin{align}
    q_a &\leq 1 < q_b\\
    \iff \dfrac{K-b}{b^\alpha(b+1)N} &< \rho \leq \dfrac{K-a}{a^\alpha(a+1)N}. \label{case_c}
\end{align}

Under this condition, the corner points can be classified into three sets. First, for any type $i \leq a$ we have the corner points: 
\begin{align}
    \left(y_i = 1, \left(y_{\tilde{t}}=0, \forall \tilde{t} \neq i\right)\right), \quad  \forall i \in \{1, .. , a\},
\end{align}
where constraint \eqref{s3} is limiting the feasible region. As in Case B, we can eliminate any type $i<a$ as type $a$ is the maximizer for the objective function.

For the second set with types $j \geq b$, we have the corner points:
 \begin{equation}
    \left( y_j = \dfrac{K(j+1)}{N c_j (j+1)+ j(K+1)},\left(y_{\tilde{t}}=0, \forall \tilde{t} \neq j\right) \right), \quad \forall j \in \{b, .. , K\},
 \end{equation} 
 where constraint \eqref{s2} is limiting the feasible region for this set. The optimization problem for this case is:
 \begin{align}\label{ytcase3}
 &\max_{j \in \{b,...,K\}
 } \quad \dfrac{j}{(j+1)} \dfrac{K(j+1)}{N c_j (j+1)+ j(K+1)}.
\end{align}
As in Case A, the optimal caching type $j$, within this set, is decided according to the following range of $\alpha$:
\begin{equation}\label{eq:joptimal}
    \sigma_{j} \leq \alpha < \sigma_{j-1},
\end{equation}
when $j > b$, and type $b$ is optimal when
\begin{equation}\label{eq:boptimal}
\alpha \geq \sigma_{b}.
\end{equation}
 
The third set consists of the intersection points of the form $\left(y_i=\tilde{y_i}>0, y_j=\tilde{y_j}>0,\left(y_k=0, \forall k \notin\{i,j\}\right)\right)$ (for brevity, we call such point $(\tilde{y_i},\tilde{y_j})$) between the feasible regions for constraints \eqref{s2} and \eqref{s3}. We simplify the problem of characterizing this last set through the following two claims.

\begin{claim}\label{cl2}
For the intersection points $(\tilde{y}_i, \tilde{y}_j)$ where $i\leq a$ is fixed, and we can vary $j\geq b$, the maximum of the objective function satisfies $j=b$.
\end{claim}
\begin{IEEEproof}
The proof follows from observing that $q_i \leq 1$, $q_j>1$, the considered intersection points tightly satisfy \eqref{s2} and \eqref{s3}, and the fact that $q_j$ is monotonically increasing as $j$ increases, which implies that the value of $\tilde{y}_j$ is monotonically decreasing as $j$ increases, as the constraint \eqref{s2} is tightly met.
A straightforward objective first derivative computation would show that the statement holds. Details are omitted for brevity.

\end{IEEEproof}

\begin{claim}\label{cl3}
For the intersection points $(\tilde{y}_i, \tilde{y}_j)$ where $j\geq b$ is fixed, and we can vary $i\leq a$, the maximum of the objective function satisfies $i=a$.
\end{claim}
\begin{IEEEproof}
The proof follows in an analogous fashion to the proof of Claim \ref{cl2}, and is omitted for brevity.
\end{IEEEproof}
From Claim~\ref{cl2} and \ref{cl3}, we can eliminate all intersections except the corner point $(\tilde{y}_a, \tilde{y}_{b})$. Next, we compare the corner point $(\tilde{y}_a, \tilde{y}_{b})$ with $\left(y_a = 1, \left(y_t=0, \forall t \neq a\right)\right) $ and $\left(y_b = 1/q_b, \left(y_{\tilde{t}}=0, \forall \tilde{t} \neq b\right)\right)$ to find when each corner point maximizes our objective function.

\begin{claim}
The corner point $(\tilde{y}_a, \tilde{y}_{b})$ is the maximizer compared with the corner point $\left(y_a = 1, \left(y_t=0, \forall t \neq a\right)\right)$.
\end{claim}

\begin{proof}
Recall that constraint \eqref{s2} intersects with constraint \eqref{s3} at $(\tilde{y_a}, \tilde{y_b})$.
Now, we evaluate the two candidate solutions at our objective function in \eqref{s1} to find which one maximizes it. Starting with the point $\left(y_a = 1, \left(y_t=0, \forall t \neq a\right)\right)$, and using $f(\tilde{y}_i,\tilde{y_j})$ to denote the objective function at point $(\tilde{y}_i,\tilde{y}_j)$,
\begin{align}
    f(1,0) & = \dfrac{a}{a+1}\\
    & = \dfrac{a}{a+1}y_a + \dfrac{a}{a+1}(1-y_a)\\
    & < \dfrac{a}{a+1}y_a + \dfrac{a+1}{a+2}(1-y_a)\\
     & = f(\tilde{y}_a, \tilde{y}_{a+1}).
\end{align}
which makes $(\tilde{y}_a, \tilde{y}_{a+1})$ the maximizer of the objective function among these two candidate solutions for any $\rho$ and $\alpha$. 
\end{proof}

Finally, we compare the corner point $(\tilde{y}_a, \tilde{y}_{b})$ with $\left(y_b = 1/q_b, \left(y_{\tilde{t}}=0, \forall \tilde{t} \neq b\right)\right)$ at our objective function \eqref{s1}. 
We also use $f(\tilde{y}_i,\tilde{y_j})$ here to denote the objective function:
\begin{flalign}
    f(\tilde{y}_a, \tilde{y}_{a+1})  &\gtrless f\left(0 , \frac{1}{q_{a+1}}\right) \\
    \Leftrightarrow \left(\dfrac{a}{a+1}\right) \tilde{y}_a + \left(\dfrac{{a+1}}{a+2}\right) \tilde{y}_{a+1}    &\gtrless \left(\dfrac{{a+1}}{a+2}\right) \dfrac{1}{q_{a+1}}\\
   \Leftrightarrow \left(\dfrac{a(a+2)}{(a+1)(a+1)}\right)\dfrac{q_{a+1} - 1}{q_{a+1} - q_a}    &\gtrless \dfrac{}{} \left( \dfrac{1}{q_{a+1}} - \dfrac{1-q_a}{q_{a+1}-q_a}\right) \\
    \Leftrightarrow q_{a+1} -1   &\gtrless \left(\dfrac{(a+1)^2}{a(a+2)}\right)\dfrac{q_a(q_{a+1}-1)}{q_{a+1}}\\
   \Leftrightarrow \left(\dfrac{a}{a+1}\right)\dfrac{1}{q_a}  &\gtrless \left(\dfrac{a+1}{a+2}\right)\dfrac{1}{q_{a+1}} \\
    \Leftrightarrow \alpha   &\gtrless 
    \log_{\frac{a+1}{a}} \left( \dfrac{a+1}{a+2} \right) +1.
\end{flalign}

It follows that the condition for $(\tilde{y}_a, \tilde{y}_{b})$ to be the optimal caching decision is:
\begin{align}
    \alpha \geq \log_{\frac{a+1}{a}} \left( \dfrac{a+1}{a+2} \right)+1=\sigma_a.
\end{align}
Hence, it follows from \eqref{eq:boptimal} that $\left(y_b = 1/q_b, \left(y_{\tilde{t}}=0, \forall \tilde{t} \neq b\right)\right)$ is optimal when $\sigma_{b} \leq \alpha < \sigma_{a}$, and more generally, it follows from \eqref{eq:joptimal} that $\left(y_j = 1/q_j, \left(y_{\tilde{t}}=0, \forall \tilde{t} \neq j\right)\right)$ is optimal when $\sigma_{j} \leq \alpha < \sigma_{j-1}$, for every $j \geq b$.
The optimal solution can be concluded for \eqref{case_c} as:
\begin{equation*}
    (y^*_i,y^*_j) =
    \begin{cases*}
    \left(\dfrac{q_{a+1} - 1}{q_{a+1} - q_a},\dfrac{1-q_a}{q_{a+1} - q_a}\right),    
    \quad  \text{if } \alpha \geq \sigma_{a},\\
    \left(0 , \dfrac{1}{q_j}\right), 
        \quad  \text{if } \sigma_{j} \leq \alpha < \sigma_{j-1}, \quad j\in\{b,\cdots,K\}.
    \end{cases*}
\end{equation*}

By characterizing the solution for all cases of the constraints, we have completed the proof of Theorem \ref{the1}.

\begin{corollary}
For the worst case scenario, uncoded delivery is optimal with caching decision $t=K$, for any given $\rho$, if:
\begin{align}
    \alpha \leq \log_{\frac{K}{K-1}} \left( \dfrac{K}{K+1} \right)+1.
\end{align}
\end{corollary}
Interestingly, for any $\rho > 0$, the optimality of uncoded delivery depends only on the network architecture. That makes coded caching ineffectual when the network architectures during the placement and delivery phases are the same, as uncoded delivery is optimal when $\alpha=0$.
\section{Numerical Results}
\begin{figure}[!ht]
    \centering
    \includegraphics[width=0.7\textwidth]{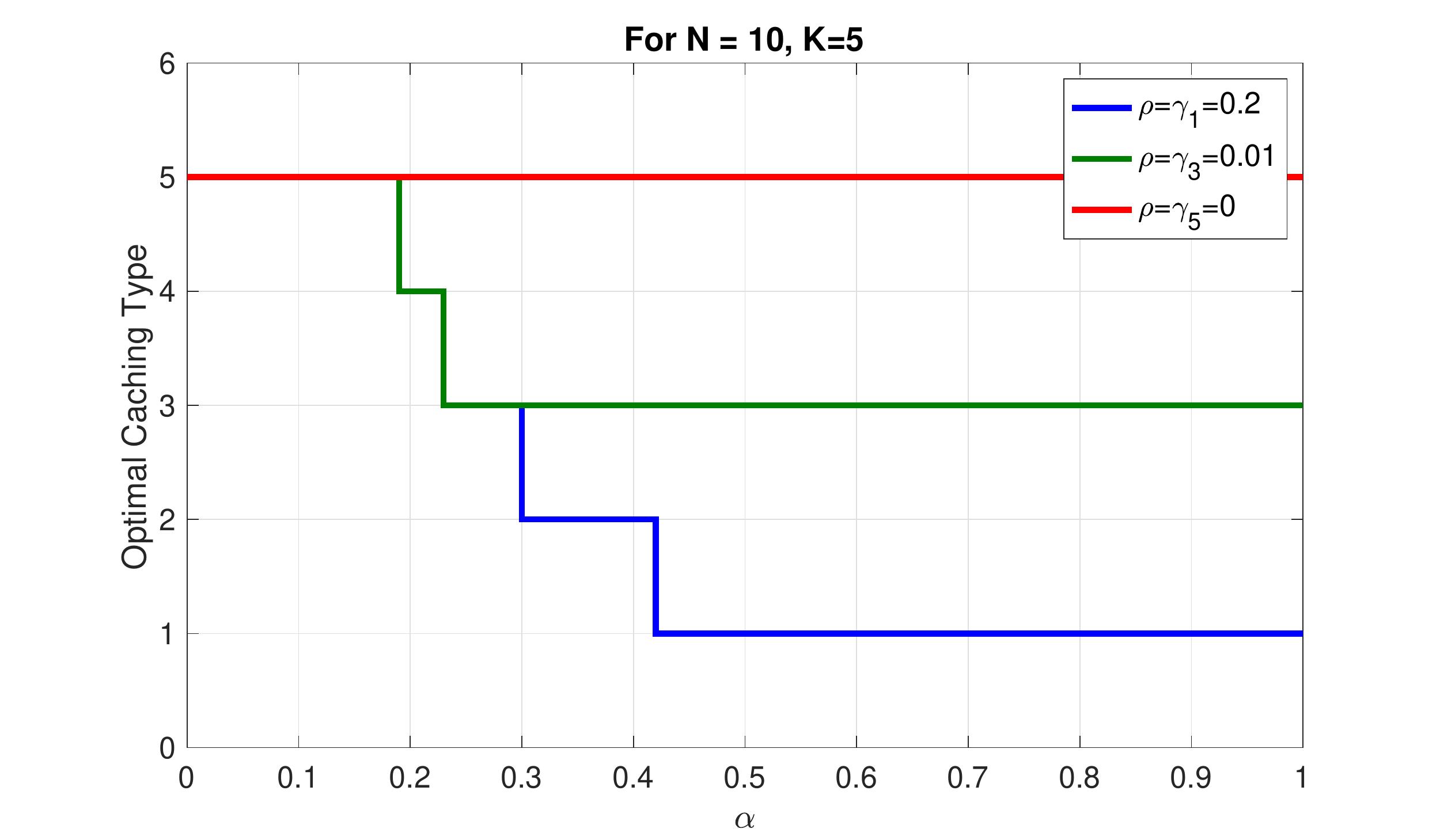}
    \caption{Effect of $\alpha$ and $\rho$  on the optimal caching type. }
    \label{fgr:costtype}
\end{figure}
\begin{figure}[ht]
    \centering
    \includegraphics[width=.7\textwidth]{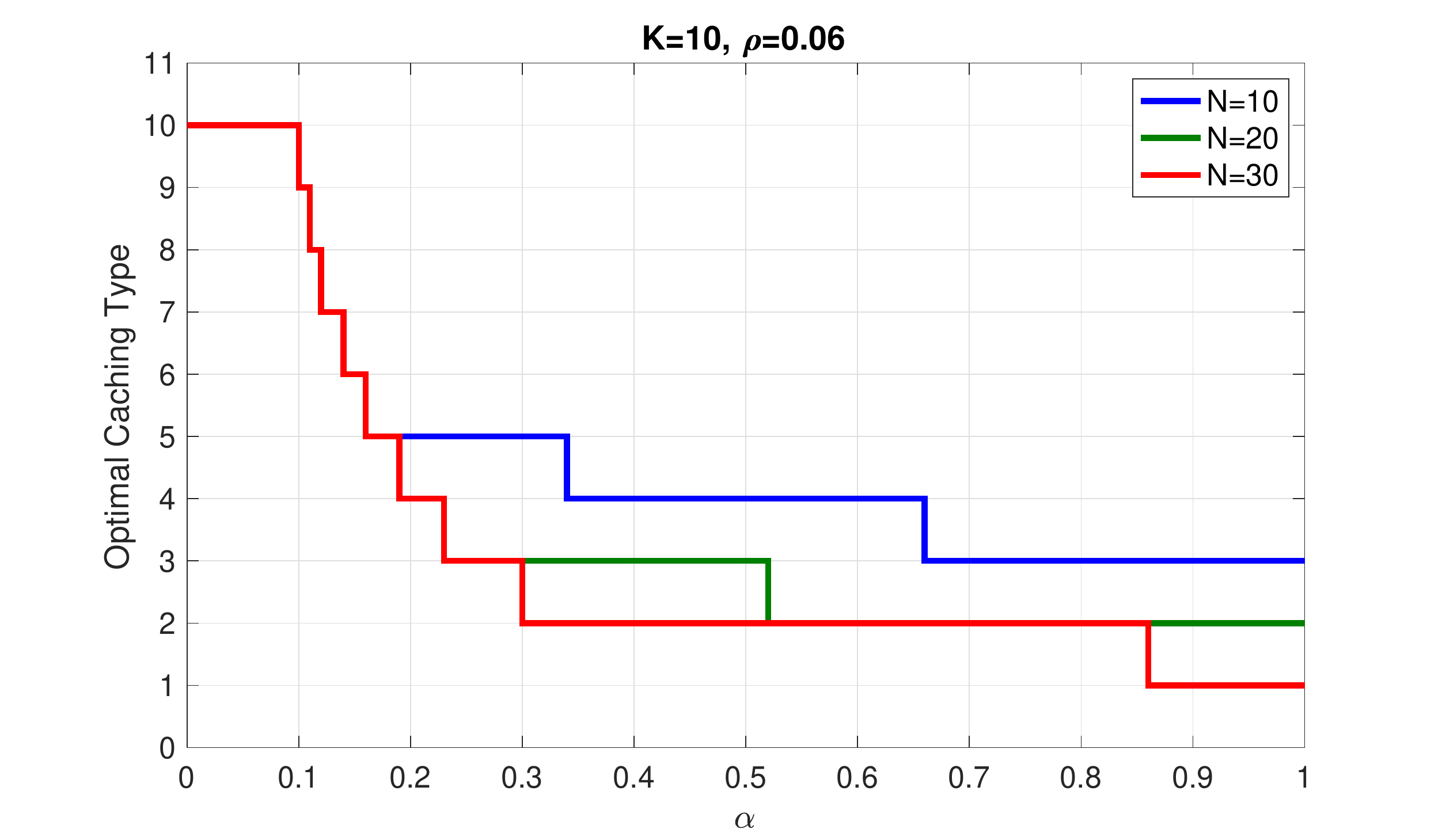}
    \caption{Effect of $N$ and $\alpha$  on the optimal caching type. }
    \label{fgr:costtype2}
\end{figure}

\begin{figure}[ht]
    \centering
    \includegraphics[width=.7\textwidth]{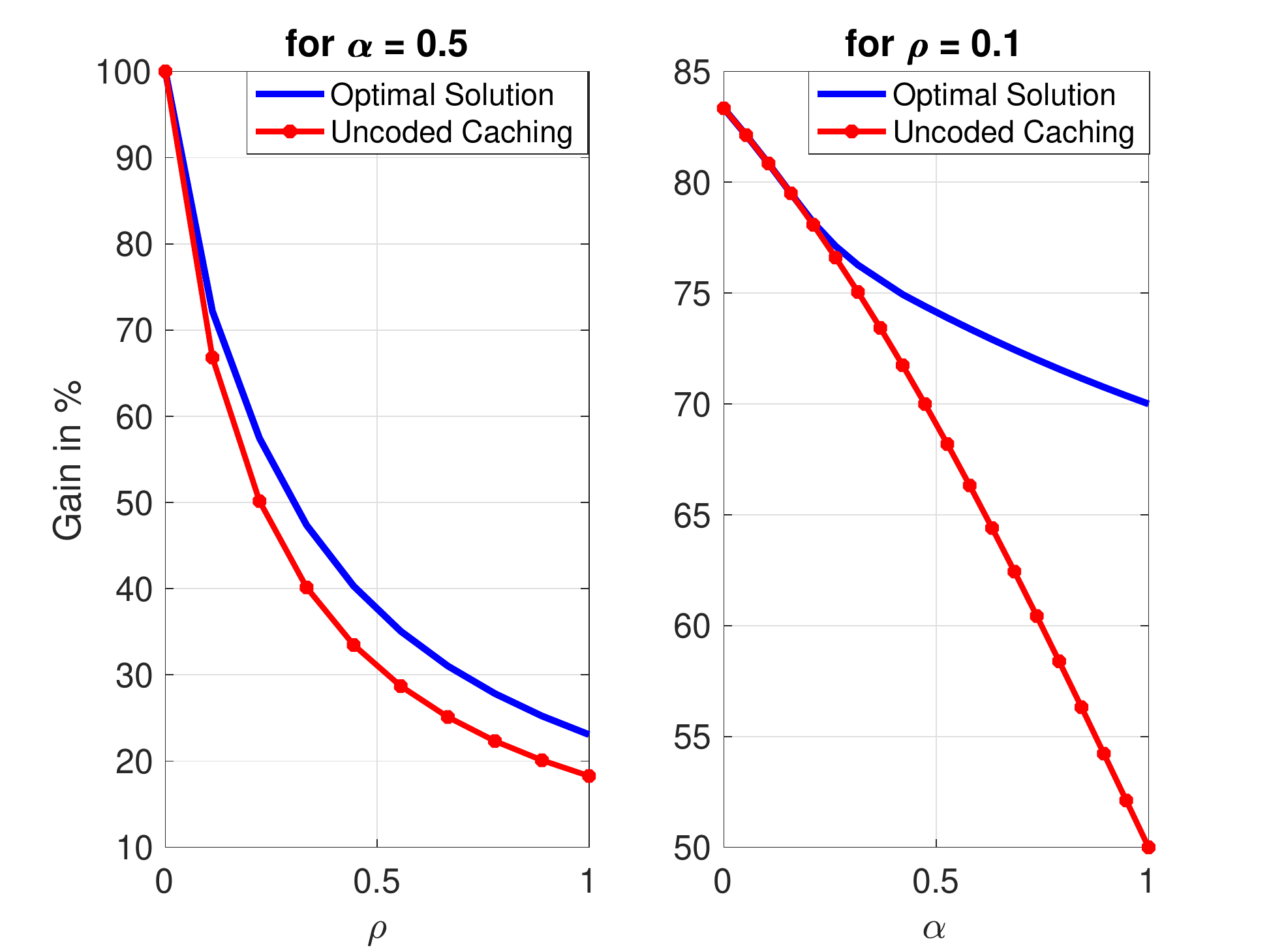}
    \caption{Effect of $\alpha$ and $\rho$  on the caching gain. }
    \label{fgr:costtype3}
\end{figure}
We first consider a system consisting of $K=5$ users interested in a library of $N =10$ files. Based on the results in Theorem \ref{the1}, we show in Figure~\ref{fgr:costtype} how the optimal caching type, i.e., how many users cache the same subfile, changes with different settings of of the cost parameter $\alpha$ (i.e., different network architectures during the placement phase) and the linear cost multiplier $\rho$.  As $\alpha$ increases,  the optimal caching type $t^*$ decreases - in a step-wise fashion - when $\rho > \gamma_{t^*}$. Note that the free placement regime corresponds to the case when $\rho = \gamma_K = 0$.
In Figure~\ref{fgr:costtype2}, we show how the optimal caching type changes with the size of the service provider's library ($N$). The optimal decision for a given $\rho$ depends on the values of $\gamma_t$, that are inversely proportional to $N$. Hence, for a large value of $\alpha$, in the worst case scenario, the optimal caching type decreases as the number of files increases. Finally, in Figure \ref{fgr:costtype3}, we show how the gain obtained through our characterized optimal solution diverges from that of uncoded caching as the placement cost increases.

\newpage
\section{Conclusion}
In this work, we derived a rigorous result characterizing the impact of placement cost on the caching gain and structure of optimal caching policies under a general framework that allows for varying the network architecture and cost per transmission across the placement and delivery phases. Future work can extend this result by adopting a stochastic model for traffic patterns, user mobility, and available end user memory.


\bibliographystyle{IEEEtran}
\bibliography{ref.bib}

\end{document}